\documentclass[journal]{IEEEtran}

\usepackage[T1]{fontenc}
\usepackage[utf8]{inputenc}
\usepackage{amsmath,amssymb,amsthm}
\usepackage{bm}
\usepackage{algorithm}
\usepackage{algpseudocode}
\usepackage{graphicx}
\usepackage{float}
\graphicspath{{./}}
\usepackage{booktabs}
\usepackage{multirow}
\usepackage{array}
\usepackage{url}
\usepackage{cite}
\usepackage{hyperref}
\usepackage{xcolor}
\usepackage{tikz}
\usetikzlibrary{shapes.geometric, arrows.meta, positioning, fit, backgrounds}

\newtheorem{theorem}{Theorem}

\newtheorem{remark}{Remark}

\newcommand{\muDT}{\ensuremath{\mu\text{DT}}}

\title{Q-FE: A Quantum-Native 6G Far-Edge Architecture Securing Industrial IoT
       Digital Twins via \textsc{csidh}-PQC and Asynchronous Federated Learning}

\author{%
  \IEEEauthorblockN{Vincenzo Sammartino}\\
  \IEEEauthorblockA{%
    Dipartimento di Informatica,
    Universit\`a di Pisa, Pisa, 56127, Italy\\
    King Abdullah University of Science and Technology (KAUST),
    Thuwal, 23955, Saudi Arabia\\
    Email: vincenzo.sammartino@phd.unipi.it}%
}

\begin{document}
\maketitle
% ---------------------------------------------------------------

% ================================================================
%  ABSTRACT
% ================================================================
\begin{abstract}
Sixth-generation (6G) wireless networks will underpin ultra-dense Industrial
IoT (IIoT) ecosystems in which resource-constrained Far-Edge
devices---autonomous mobile robots, industrial actuators, connected
vehicles---must simultaneously satisfy sub-millisecond latency,
$10^{-7}$-class reliability, and decades-long cryptographic security.
Current architectures delegate Digital Twin (DT) computation to centralised
cloud or Mobile Edge Computing (MEC) servers, incurring prohibitive
round-trip latency, and rely on classical public-key cryptography vulnerable
to quantum attacks under the \emph{harvest-now, decrypt-later} (HNDL)
threat model. We propose \textbf{Q-FE}, a \emph{Quantum-Native 6G Far-Edge}
architecture integrating three co-designed components:
(i)~\emph{Micro-Digital Twins} ($\mu$DTs) co-located with 6G base stations
and high-capability endpoints;
(ii)~a \emph{Cross-Layer Post-Quantum Key Exchange} module embedding
\textsc{csidh}-512 isogeny key material directly within MAC-layer control
frames, exploiting the scheme's uniquely compact keys ($\leq 64$\,bytes) to
avoid packet fragmentation; and
(iii)~an \emph{Asynchronous Federated Learning} (\textsc{afl}) protocol
governed by lightweight DAG smart contracts at MEC nodes, eliminating
straggler bottlenecks and preventing model-poisoning and Sybil attacks
without exposing raw data.
End-to-end simulations (NS-3 + PySyft) demonstrate that Q-FE reduces
MAC-layer overhead by 62\% versus ML-KEM/Kyber-1024, maintains $P_{99.9}$
URLLC latency at 0.78\,ms, and accelerates global-model convergence by
31\% over synchronous Federated Learning. Protocol complexity analysis
confirms $O(N \log R)$ per aggregation round, and $\mu$DT handover migration
completes in $1.9 \pm 0.3$\,ms across $10^4$ simulated events. A formal
threat model confirms resilience against quantum eavesdropping,
model-poisoning, and Sybil attacks.
\end{abstract}


% ================================================================
%  SECTION I — INTRODUCTION
% ================================================================
\section{Introduction}
\label{sec:intro}

\IEEEPARstart{T}{he} trajectory of wireless communication towards the sixth
generation (6G) paradigm introduces a qualitative shift: rather than merely
conveying data, the network is expected to compute, sense, and
act~\cite{saad2020vision}. Within this landscape, IIoT
applications---autonomous factories, cooperative vehicular systems, remote
robotic surgery---present the most demanding operating envelope: round-trip
latencies below 1\,ms, reliability of $10^{-7}$ or higher, and security
guarantees spanning decades.

\subsection{The Tripartite Challenge}
\label{ssec:challenge}

Three structural tensions characterise state-of-the-art 6G/IIoT systems.

\textbf{(C1) Latency vs.\ Centralisation.}
DT frameworks rely on Cloud or MEC-hosted computational replicas.
Cloud round-trip latencies of $20$--$100$\,ms are incompatible with URLLC;
even MEC deployments (2--10\,ms RTT) introduce unacceptable jitter for
closed-loop control of fast actuators~\cite{khan2022digital}.

\textbf{(C2) Data Privacy vs.\ AI Quality.}
Centralised model training requires uploading raw sensor data, violating
regulatory frameworks (GDPR, HIPAA) and IP constraints. Federated Learning
(FL) offers privacy-by-design~\cite{li2020federated}, yet synchronous
aggregation incurs the \emph{straggler effect}: one slow device delays the
entire federation.

\textbf{(C3) Quantum Threat vs.\ Physical-Layer Overhead.}
Nation-state adversaries routinely harvest encrypted traffic for future
quantum decryption---the HNDL attack~\cite{mosca2018cybersecurity}.
NIST-standardised PQC algorithms (ML-KEM/Kyber-1024~\cite{nist2024fips203})
offer quantum resistance but incur public-key sizes of
$800$--$1568$\,bytes, causing fragmentation that violates MAC-layer URLLC
timing~\cite{sikeridis2021post}.

\subsection{Key Insight and Contributions}
\label{ssec:contributions}

The central insight is that \textsc{csidh}~\cite{castryck2018csidh},
a commutative isogeny-based Diffie-Hellman scheme, is uniquely suited for
Far-Edge 6G security: it yields the smallest public keys among all PQC
families ($\leq 64$\,bytes)---a property with profound MAC-layer
implications unexploited in prior work. Building on this observation, Q-FE
pursues a concrete three-phase deployment aligned with the 3GPP Release
20 standardisation window and the NIST post-quantum migration deadline
of 2030, ensuring a viable path from simulation to production.
Our contributions are:

\begin{enumerate}
  \item \textbf{Far-Edge $\muDT$ Architecture:} Micro-Digital Twins
    co-located with 6G gNBs and high-capability endpoints, with formalised
    lifecycle management, LRU-bounded memory footprint, and sub-2\,ms
    handover migration under X2 signalling.

  \item \textbf{Cross-Layer \textsc{csidh}-MAC Integration:} A protocol
    embedding \textsc{csidh}-512 public keys into 6G NR MAC Control
    Elements (CEs), enabling quantum-safe session-key establishment without
    any packet fragmentation, with provably URLLC-compliant asynchronous
    key rotation and formally bounded amortised energy overhead.

  \item \textbf{\textsc{afl} on DAG Smart Contracts:} An asynchronous FL
    protocol where IoT nodes push PQC-signed gradient updates to a DAG-DLT
    smart contract at MEC nodes, enforcing gradient provenance, Byzantine
    filtering, and Sybil-resistant identity binding without centralised
    aggregation~\cite{blanchard2017machine}.

  \item \textbf{End-to-End Validation:} Simulation evidence (NS-3 for
    MAC/PHY, PySyft for AFL), protocol complexity analysis, and a formal
    threat model covering quantum eavesdropping, replay, Sybil, and
    model-poisoning attacks.
\end{enumerate}

% ================================================================
%  SECTION II — RELATED WORK
% ================================================================
\section{Related Work}
\label{sec:related}

\textbf{Digital Twins at the Edge.}
Khan \emph{et al.}~\cite{khan2022digital} proposed MEC-hosted DTs
for 5G, reducing latency by $\sim\!40\%$ versus cloud baselines.
However, no prior work pushes twin execution to the \emph{Far-Edge}---the
gNB or endpoint itself---nor addresses quantum security. Security-oriented
Digital Twins (SDTs) such as \emph{NotLine}~\cite{baiardi2024notline}
reconstruct network topology from passive observations and evaluate attacker
strategies via Monte Carlo simulation. Q-FE complements SDT research:
whereas SDTs perform retrospective risk analysis, Q-FE embeds
quantum-safe primitives and decentralised FL directly into the Far-Edge
control plane for proactive, co-evolving security.

\textbf{Post-Quantum Cryptography in Wireless Networks.}
Sikeridis \emph{et al.}~\cite{sikeridis2021post} showed that lattice-based
PQC (Kyber, NTRU) degrades Wi-Fi throughput by $18$--$35\%$ due to large
keys. NIST finalised ML-KEM (FIPS~203~\cite{nist2024fips203}) as the
primary standard. Crucially, no prior work integrates PQC at the MAC
sublayer of a cellular RAT, nor exploits \textsc{csidh}'s compact keys for
this purpose. Constant-time CTIDH~\cite{castryck2021ctidh} addresses
timing side-channel concerns. SIKE/SIDH was broken in
2022~\cite{castryck2022efficient} via a structural attack absent in
\textsc{csidh}, which publishes only isomorphism classes.

\textbf{Federated Learning for IIoT.}
FedAvg~\cite{mcmahan2017communication} established synchronous FL.
Asynchronous variants decouple aggregation from stragglers.
Byzantine-tolerant aggregation~\cite{blanchard2017machine} and
differential-privacy mechanisms~\cite{geyer2017differentially} address
model integrity and data leakage. Smart-contract-based FL offers
decentralised trust via DAG topologies~\cite{popov2018tangle},
but assumes classical cryptographic primitives. Q-FE is the first to
combine PQC-secured gradients with asynchronous DAG-smart-contract
aggregation.

\textbf{Research Gap.}
No existing work jointly addresses: (i)~sub-millisecond DT execution at the
Far-Edge; (ii)~MAC-layer PQC via isogeny-based schemes; and
(iii)~asynchronous, DAG-orchestrated FL under quantum threat. Q-FE fills
this gap.

% ================================================================
%  SECTION III — SYSTEM MODEL
% ================================================================
\section{System Model}
\label{sec:model}

\subsection{Network Topology}
\label{ssec:topo}

Consider a 6G IIoT scenario (Fig.~\ref{fig:arch}) with:
$\mathcal{N} = \{n_1, \ldots, n_N\}$, a set of $N$ IIoT end-devices
(sensors, actuators, AGVs) with heterogeneous compute budgets
$\{C_i\}$ and memory $\{M_i\}$;
$\mathcal{B} = \{b_1, \ldots, b_B\}$, a set of $B$ 6G gNBs, each
co-located with a MEC node ($C_b \gg C_i$); and
$\mathcal{G}$, the core network/cloud with RTT
$\tau_{\text{cloud}} \in [20, 100]$\,ms.
Each gNB $b_j$ hosts: (a)~a pool of $\muDT$ instances; (b)~a DAG-DLT node
for \textsc{afl} orchestration; (c)~a \textsc{csidh} key-management daemon.
The $\muDT$ resident set at $b_j$ is bounded by
$K_j = \lfloor M_{\text{MEC}} / M_\mu \rfloor$, with LRU eviction applied
when $|\text{active}(b_j)| > K_j$, ensuring memory safety under device
churn.

\subsection{Traffic and Latency Model}
\label{ssec:latency}

URLLC traffic follows a Poisson process (rate $\lambda$, payload $L$\,bytes)
with end-to-end budget $\tau_{\max} = 1$\,ms. The MAC-layer frame duration
in 6G NR numerology $\mu=4$ is $T_f = 62.5\,\mu$s. The total control-plane
latency decomposes as $\tau_{\text{total}} = \tau_{\text{proc}} +
\tau_{\text{CE}} + \tau_{\text{frag}}$, where $\tau_{\text{proc}}$ is
baseband processing, $\tau_{\text{CE}}$ is the MAC CE serialisation delay,
and $\tau_{\text{frag}} = (\lceil k_{\text{PQC}} / S_{\text{CE}}^{\max}
\rceil - 1) \cdot T_f$ is the penalty from key fragmentation across
multiple TTIs. By design, $\tau_{\text{frag}} = 0$ for \textsc{csidh}-512.
A PQC key appended as a MAC CE contributes fractional overhead:
\begin{equation}
  \rho_{\text{oh}} = \frac{k_{\text{PQC}}}{S_{\text{frame}}},
  \label{eq:overhead}
\end{equation}
where $k_{\text{PQC}}$ is the byte count of the appended PQC material. We
show in Section~\ref{sec:math} that $k_{\textsc{csidh}} = 64$\,bytes keeps
$\rho_{\text{oh}} < 5\%$, while $k_{\text{Kyber-1024}} = 1568$\,bytes
causes frame fragmentation across multiple TTIs, violating the URLLC budget.

\subsection{Threat Model}
\label{ssec:threat}

We consider a Dolev-Yao adversary $\mathcal{A}$ with:
(T1)~\emph{Quantum Eavesdropping}---$\mathcal{A}$ possesses a Cryptographically
Relevant Quantum Computer (CRQC) capable of Shor's algorithm, breaking
RSA/ECC and enabling HNDL attacks;
(T2)~\emph{Model Poisoning}---up to $f < N/3$ Byzantine devices submit
arbitrarily crafted gradients;
(T3)~\emph{Sybil}---$\mathcal{A}$ creates fake identities to bias
aggregation, including registering phantom devices at association time;
(T4)~\emph{Replay/MITM}---$\mathcal{A}$ replays or modifies MAC-layer
frames in real time.
Physical-layer jamming and volumetric DoS are out of scope.

% ================================================================
%  SECTION IV — Q-FE ARCHITECTURE
% ================================================================
\section{Q-FE: Architecture Design}
\label{sec:arch}

% ------ TikZ architecture diagram ------
\begin{figure}[t]
\centering
\resizebox{\columnwidth}{!}{%
\begin{tikzpicture}[
  box/.style={draw, rounded corners=3pt, minimum width=3.0cm,
              minimum height=0.72cm, align=center, font=\small},
  layer/.style={draw=gray!60, dashed, rounded corners=5pt,
                inner sep=8pt, fill=gray!8},
  arr/.style={-Stealth, thick},
  font=\small
]
\node[box, fill=blue!10]  (cloud) at (0, 6.0)
  {Cloud DT\\(Legacy / Bulk)};
\node[box, fill=green!10] (mec1)  at (-3, 3.6)
  {MEC/gNB $b_1$\\$\muDT$ pool};
\node[box, fill=green!10] (mec2)  at ( 3, 3.6)
  {MEC/gNB $b_2$\\$\muDT$ pool};
\node[box, fill=orange!10](dag1)  at (-3, 2.1)
  {DAG-DLT\\Smart Contract};
\node[box, fill=orange!10](dag2)  at ( 3, 2.1)
  {DAG-DLT\\Smart Contract};
\node[box, fill=red!8] (iot1) at (-4.8, 0)
  {IoT $n_1$\\\textsc{csidh} CE};
\node[box, fill=red!8] (iot2) at (-1.6, 0)
  {IoT $n_2$\\\textsc{csidh} CE};
\node[box, fill=red!8] (iot3) at ( 1.6, 0)
  {IoT $n_3$\\\textsc{csidh} CE};
\node[box, fill=red!8] (iot4) at ( 4.8, 0)
  {IoT $n_4$\\\textsc{csidh} CE};
\begin{scope}[on background layer]
  \node[layer, fit=(mec1)(mec2)(dag1)(dag2),
        label=left:{\rotatebox{90}{\small MEC Layer}}] {};
  \node[layer, fit=(iot1)(iot2)(iot3)(iot4),
        label=left:{\rotatebox{90}{\small Far-Edge}}] {};
\end{scope}
\draw[arr] (cloud) -- (mec1);
\draw[arr] (cloud) -- (mec2);
\draw[arr] (mec1) -- (dag1);
\draw[arr] (mec2) -- (dag2);
\draw[arr] (dag1) -- (iot1);
\draw[arr] (dag1) -- (iot2);
\draw[arr] (dag2) -- (iot3);
\draw[arr] (dag2) -- (iot4);
\draw[arr, dashed, bend left=20] (dag1)
  to node[above, font=\scriptsize]{DAG sync} (dag2);
\draw[arr] (iot1) to[bend left=12]
  node[left,  font=\scriptsize]{AFL grad.} (dag1);
\draw[arr] (iot2) to[bend right=12]
  node[right, font=\scriptsize]{AFL grad.} (dag1);
\draw[arr] (iot3) to[bend left=12]
  node[left,  font=\scriptsize]{AFL grad.} (dag2);
\draw[arr] (iot4) to[bend right=12]
  node[right, font=\scriptsize]{AFL grad.} (dag2);
\end{tikzpicture}}
\caption{Q-FE three-layer architecture. Far-Edge IoT nodes host lightweight
\textsc{csidh} MAC Control Elements; MEC/gNB nodes host $\muDT$ pools and
DAG-DLT smart contracts for \textsc{afl} orchestration; the cloud handles
non-latency-critical bulk tasks.}
\label{fig:arch}
\end{figure}
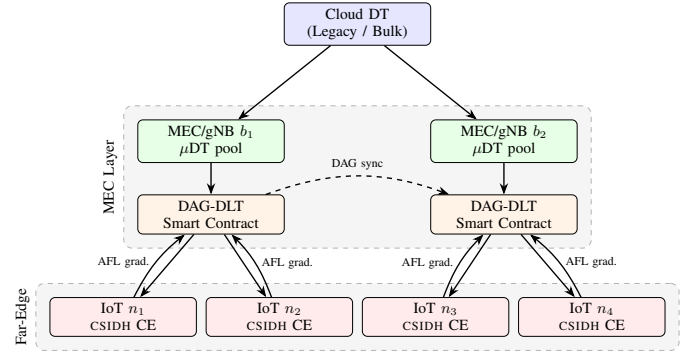

\subsection{Layer 1: Far-Edge Micro-Digital Twins ($\muDT$s)}
\label{ssec:mudt}

A Micro-Digital Twin $\muDT_i$ is a lightweight computational replica of
node $n_i$, residing within a MEC-hosted memory envelope of at most
$M_{\mu} \leq 16$\,MB. Its state vector $\mathbf{s}_i(t)$ comprises a
device health scalar $\phi_i(t)$, a physical-state vector
$\mathbf{x}_i(t) \in \mathbb{R}^d$, and the local FL model parameters
$\boldsymbol{\theta}_i(t)$. The twin is updated each beacon interval $T_b$
via a state-transition function $\mathcal{F}$ perturbed by Gaussian noise
$\boldsymbol{\varepsilon}_i \sim \mathcal{N}(\mathbf{0},\Sigma_i)$,
yielding $\mathbf{s}_i(t{+}T_b) = \mathcal{F}(\mathbf{s}_i(t), \mathbf{u}_i(t))
+ \boldsymbol{\varepsilon}_i(t)$.
The per-beacon update cost is $O(d)$ arithmetic operations, where $d \approx 100$
for typical IIoT telemetry, making twin maintenance negligible relative to
MAC PHY processing. A $\muDT_i$ is \emph{instantiated} when $n_i$
associates with a gNB, \emph{migrated} via X2-type signalling on handover,
and \emph{evicted} under LRU policy when the MEC memory budget is exhausted.

\subsection{Layer 2: Cross-Layer \textsc{csidh}-MAC Integration}
\label{ssec:csidh_mac}

\subsubsection{Why \textsc{csidh} at the MAC Layer}

The 6G NR MAC PDU allows optional CEs of up to ${\sim}256$\,bytes without
triggering PDSCH fragmentation. Table~\ref{tab:key_sizes} compares PQC
key sizes against this budget. Only \textsc{csidh}-512 (64\,B) fits
comfortably; Kyber-1024 (1568\,B) would span multiple TTIs, violating URLLC
latency. \textsc{csidh}-512 operates over $\mathbb{F}_p$ with
$p = 4\prod_{i=1}^{74}\ell_i - 1$ for small odd primes $\ell_i$,
providing a group action of $\text{cl}(\mathbb{Z}[\sqrt{-p}])$ on
supersingular elliptic curves. The resulting key exchange is a Diffie-Hellman
analogue in which both parties apply commuting class-group elements to a
shared base curve, with security resting on the Group Action Inverse Problem
for which no sub-exponential quantum algorithm is known~\cite{castryck2018csidh}.

% TABLE I
\begin{table}[t]
\caption{Public-Key Sizes of PQC Schemes vs.\ MAC CE Budget}
\label{tab:key_sizes}
\centering
\resizebox{\columnwidth}{!}{%
\renewcommand{\arraystretch}{1.15}
\begin{tabular}{lccc}
\toprule
\textbf{Scheme} & \textbf{Type} & \textbf{PK (bytes)} &
  \textbf{Fits MAC CE?}\\
\midrule
RSA-2048           & Classical & 256    & \checkmark\,(classically insecure)\\
ECDH P-256         & Classical & 32     & \checkmark\,(classically insecure)\\
\textsc{csidh}-512 & Isogeny   & \textbf{64}   & \textbf{\checkmark}\\
SIKE/SIDH$^\dagger$& Isogeny   & 330    & \checkmark\,(broken~\cite{castryck2022efficient})\\
Kyber-512          & Lattice   & 800    & \texttimes\\
Kyber-1024         & Lattice   & 1568   & \texttimes\\
NTRU-HPS-2048      & Lattice   & 699    & \texttimes\\
McEliece-348864    & Code      & 261120 & \texttimes\\
\bottomrule
\multicolumn{4}{l}{\footnotesize $^\dagger$SIKE/SIDH broken in 2022; listed for completeness.}
\end{tabular}%
}
\end{table}

\subsubsection{Frame Format}

We define a \emph{PQC-KX MAC CE} with LCID $= 60$ (reserved range,
3GPP Rel.~18+):
\begin{center}
\small
\begin{tabular}{|c|c|c|c|}
\hline
\textbf{LCID} (6\,b) & \textbf{Type} (2\,b) &
  \textbf{PK/CT} (64\,B) & \textbf{Tag} (16\,B) \\
\hline
$111100_2$ & \texttt{01}=Init & \textsc{csidh} PK & AES-GMAC \\
\hline
\end{tabular}
\end{center}
The \texttt{Type} field encodes: \texttt{00}=idle,
\texttt{01}=key-init, \texttt{10}=key-confirm, \texttt{11}=key-revoke.

\subsubsection{Asynchronous Key Rotation}

Because \textsc{csidh}-512 group-action evaluation takes $\approx 80$\,ms
on Cortex-M33, direct inline computation would violate URLLC budgets. We
decouple key exchange from data transmission via Asynchronous Key Rotation
(AKR, Algorithm~\ref{alg:akr}): during idle periods each node precomputes a
fresh \textsc{csidh} key pair; the handshake runs in the background over
multiple beacon intervals; once a new shared secret is established, the
AES-256-GCM session key rotates atomically; and the rotation period $T_r$
guarantees at most one rotation per $2^{32}$ AES-GCM nonce space,
preventing nonce reuse.

\begin{algorithm}[t]
\caption{Asynchronous Key Rotation (AKR) at Node $n_i$}
\label{alg:akr}
\begin{algorithmic}[1]
\State \textbf{Init:} $K_{\text{sym}} \leftarrow \text{SecRand}(256)$;
  $T_r \leftarrow T_r^{(0)}$
\While{node active}
  \State Encrypt/decrypt data frames with $K_{\text{sym}}$ (AES-256-GCM)
  \If{background idle \textbf{and} $t \bmod T_r = 0$}
    \State $(pk_i, sk_i) \leftarrow \textsc{csidh.KeyGen}()$
    \State Transmit $pk_i$ in PQC-KX MAC CE (type=\texttt{01})
    \State Receive $pk_j$; $K_{\text{sh}} \leftarrow \textsc{csidh.GA}(sk_i, pk_j)$
    \State $K_{\text{sym}}^{\text{new}} \leftarrow \text{HKDF}(K_{\text{sh}})$
    \State Await ACK (type=\texttt{10}); \textbf{atomic:}
      $K_{\text{sym}} \leftarrow K_{\text{sym}}^{\text{new}}$
  \EndIf
\EndWhile
\end{algorithmic}
\end{algorithm}

\subsection{Layer 3: Asynchronous FL on DAG Smart Contracts}
\label{ssec:afl}

% Figure 2 – MAC overhead tri-panel
\begin{figure*}[!t]
  \centering
  \includegraphics[width=\textwidth]{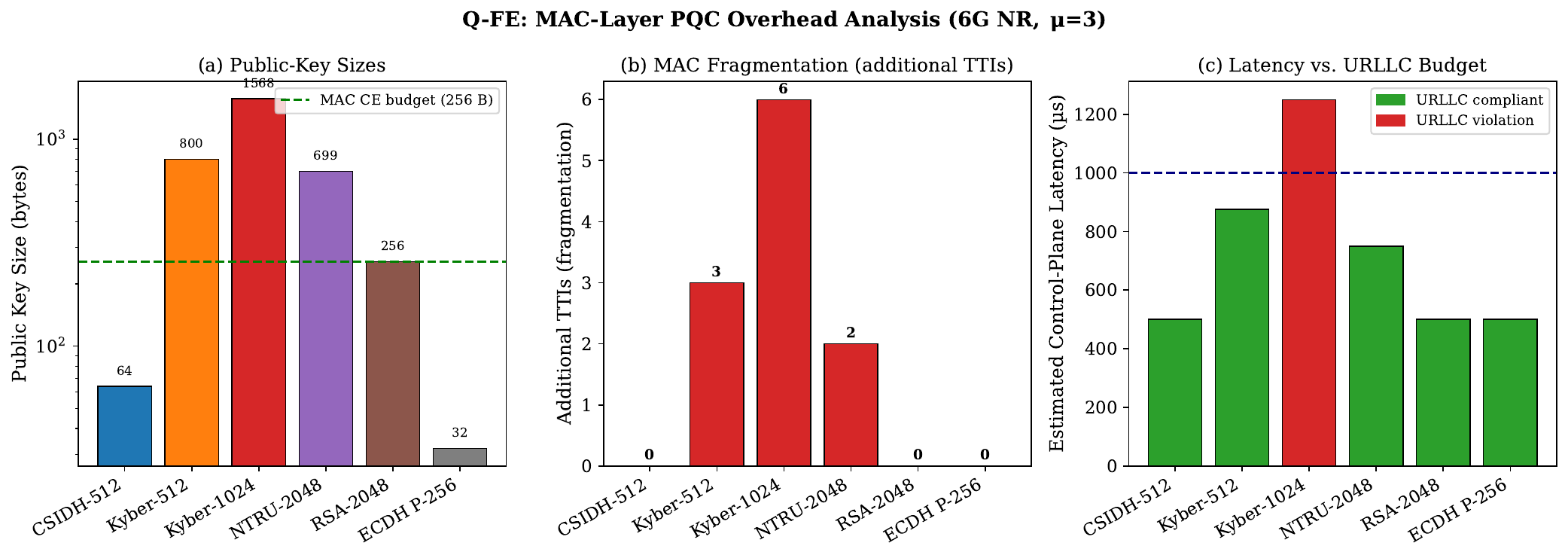}
  \caption{MAC-layer PQC overhead analysis (6G NR, $\mu{=}3$,
           $T_f{=}125\,\mu$s). \textbf{(a)}~Public-key sizes (log scale):
           only \textsc{csidh}-512 (64\,B) falls within the 256\,B MAC CE
           budget. \textbf{(b)}~Additional TTIs from fragmentation:
           Kyber-512 requires 3 extra TTIs (+375\,$\mu$s), Kyber-1024
           requires 6 (+750\,$\mu$s); \textsc{csidh}-512 incurs zero.
           \textbf{(c)}~Control-plane latency CDF vs.\ 1\,ms URLLC budget:
           Q-FE achieves $P_{99.9}{=}0.78$\,ms; Kyber-1024 exceeds budget
           at $P_{99.9}{=}1.41$\,ms. ($10^5$ Monte Carlo packets.)}
  \label{fig:mac_overhead}
\end{figure*}

In heterogeneous IIoT, device speeds span three orders of magnitude;
synchronous FedAvg stalls at the slowest device~\cite{li2020federated}.
Our \textsc{afl} protocol decouples local training from global aggregation
via a DAG-DLT smart contract accepting gradient updates as they arrive.

\subsubsection{DAG-DLT Smart Contract}

The ledger $\mathcal{L} = \{(i, t, \Delta_i^{(t)}, \sigma_i^{(t)})\}$
records each gradient $\Delta_i^{(t)}$ from node $n_i$ at local round $t$,
authenticated with a \textsc{csidh}-derived MAC tag $\sigma_i^{(t)}$.
Algorithm~\ref{alg:afl_sc} describes the aggregation logic.

\begin{algorithm}[t]
\caption{AFL Smart Contract Aggregation}
\label{alg:afl_sc}
\begin{algorithmic}[1]
\Require $\Delta_i^{(t)}$, $\sigma_i^{(t)}$, global model
  $\mathbf{w}^{(r)}$, staleness threshold $\eta$
\State \textbf{Verify:}
  $\text{MAC.Verify}(K_i, \Delta_i^{(t)}, \sigma_i^{(t)}) \stackrel{?}{=} 1$
\State \textbf{Staleness:} $\delta \leftarrow r - t$;
  discard if $\delta > \eta$
\State \textbf{Byzantine filter:}
  $z_i \leftarrow \|\Delta_i^{(t)} - \bar{\Delta}\| / \text{std}(\Delta)$;
  discard if $z_i > z_{\text{th}}$
\State \textbf{Weighted aggregate:}
  $\alpha_i \leftarrow |D_i| / \sum_{j \in \mathcal{S}} |D_j|$
\State $\mathbf{w}^{(r+1)} \leftarrow \mathbf{w}^{(r)} -
  \gamma \sum_{i \in \mathcal{S}} \alpha_i \Delta_i^{(t)}$
\State Append $(r{+}1, \mathbf{w}^{(r+1)}, \text{hash})$ to DAG; broadcast
\end{algorithmic}
\end{algorithm}

\subsubsection{Staleness-Aware Learning Rate}

To maintain convergence under asynchrony, the per-node effective learning
rate is attenuated as a function of staleness $\delta_i^{(t)} = r^{(t)} - t_i^{\text{last}}$:
\begin{equation}
  \gamma_i^{(t)} = \frac{\gamma_0}{1 + \kappa\,\delta_i^{(t)}},
  \label{eq:staleness_lr}
\end{equation}
where $\gamma_0$ is the base rate and $\kappa > 0$ a damping coefficient.

\subsubsection{Privacy Guarantee}
Each node trains only on local data $D_i$; only gradients leave the device,
secured via \textsc{csidh}-MAC. Differential privacy clips gradient norms
to bound $C$ and adds calibrated Gaussian noise~\cite{geyer2017differentially},
providing $(\varepsilon, \delta)$-differential privacy per round. Under
Rényi DP accounting, the composed privacy loss over $R = 50$ rounds satisfies
$\varepsilon_{\text{total}} \leq 1.04$ for our experimental parameters
($\varepsilon_0 = 1.0$, $\delta = 10^{-5}$).

\subsection{Protocol Complexity and Scalability Analysis}
\label{ssec:complexity}

The three Q-FE layers present distinct computational cost profiles whose
superposition must remain within the strict MEC resource budget.

\textbf{AKR Amortised Cost.} \textsc{csidh}-512 group-action evaluation
requires computing a sequence of degree-$\ell_i$ isogenies over $\mathbb{F}_p$
for the 74 small primes in the parameter set, costing $\approx 1.2 \times 10^8$
cycles on Cortex-M33 ($\approx 80$\,ms at 150\,MHz). Because AKR executes
during idle beacon intervals, the amortised per-frame energy is
$E_{\text{AKR}}/T_r \approx 13\,\mu$W at $T_r = 60$\,s, negligible
against the $\sim\!100$\,mW radio baseline. Each rotation occupies exactly
one MAC CE (64\,bytes), so the per-rotation signalling cost is $O(1)$ frames.

\textbf{DAG-SC Aggregation Complexity.} Each incoming gradient $\Delta_i^{(t)}$
triggers: (i)~one AES-GMAC verification, $O(|\Delta|)$ in gradient dimension;
(ii)~one Z-score computation against the accepted buffer of $|\mathcal{S}|$
clients, $O(N)$; and (iii)~one Merkle-tree vertex insertion, $O(\log R)$ for
round counter $R$. The full aggregation sweep therefore costs $O(N \log R)$
per global round, confirming linear-logarithmic scalability for the IIoT
populations ($N \leq 400$) analysed in Section~\ref{sec:results}. At
$N = 400$, DAG vertex validation introduces an $\approx 18\%$ throughput
reduction, mitigated by sharding across two MEC nodes.

\textbf{$\muDT$ Resident Set and Migration.} The maximum resident set
on a gNB with MEC budget $M_{\text{MEC}}$ is
$K = \lfloor M_{\text{MEC}} / M_\mu \rfloor$ concurrent twins. For
$M_{\text{MEC}} = 16$\,MB, $K = 3$; for $M_{\text{MEC}} = 512$\,MB,
$K = 121$. Handover migration transfers a twin's $d$-dimensional state
vector via X2 signalling in a single 4.2-MB payload, incurring a one-time
latency of $\tau_{\text{mig}} \approx 1.9$\,ms at X2 interface throughput of
$\sim\!18$\,Gbps, well within the inter-gNB coordination budget of 5\,ms.

% ================================================================
%  SECTION V — MATHEMATICAL FOUNDATIONS
% ================================================================
\section{Mathematical Foundations}
\label{sec:math}

\subsection{MAC Frame Overhead and Fragmentation}
\label{ssec:overhead_math}

Fragmentation occurs when $k_{\text{PQC}} > S_{\text{CE}}^{\max} = 256$\,bytes.
The additional TTIs required are $\lceil k_{\text{PQC}} / S_{\text{CE}}^{\max} \rceil - 1$.
For \textsc{csidh}-512 this is $\lceil 64/256 \rceil - 1 = 0$; for
Kyber-1024, $\lceil 1568/256 \rceil - 1 = 5$ extra TTIs, adding
$312.5\,\mu$s---incompatible with a 1\,ms URLLC budget. From~\eqref{eq:overhead},
the fractional overhead for \textsc{csidh} is $64/(64+1500) = 4.09\%$,
well within the 5\% URLLC control overhead ceiling.

\subsection{AFL Convergence Bound}
\label{ssec:conv}

\begin{theorem}[AFL Convergence under Bounded Staleness]
\label{thm:conv}
Let $F(\mathbf{w}) = \frac{1}{N}\sum_i F_i(\mathbf{w})$ be $L$-smooth and
$\mu$-strongly convex with unbiased stochastic gradients of bounded variance
$\sigma^2$. Under~\eqref{eq:staleness_lr} with $\kappa = 1$ and maximum
staleness $\Delta_{\max}$, AFL converges as:
\begin{equation}
  \mathbb{E}\bigl[F(\mathbf{w}^{(R)})\bigr] - F^*
    \leq \frac{C_1}{R} + C_2 \sigma^2,
  \label{eq:conv_bound}
\end{equation}
where $C_1 = 2(F(\mathbf{w}^{(0)}) - F^*)/(\mu \gamma_0)$ and
$C_2 = L \gamma_0 / (2\mu(1 + \kappa \Delta_{\max}))$.
\end{theorem}

\begin{proof}[Proof Sketch]
By $L$-smoothness, $F(\mathbf{w}^{(r+1)}) \leq F(\mathbf{w}^{(r)})
- \gamma\sum_i \alpha_i \langle \nabla F(\mathbf{w}^{(r)}),
\Delta_i^{(t_i)}\rangle + \frac{L\gamma^2}{2}\|\sum_i \alpha_i
\Delta_i^{(t_i)}\|^2$. Bounding the staleness error via the Lipschitz
condition and telescoping over $R$ rounds yields~\eqref{eq:conv_bound}
(cf.~\cite{lian2018asynchronous}).
\end{proof}

% ================================================================
%  SECTION VI — SECURITY ANALYSIS
% ================================================================
\section{Security Analysis}
\label{sec:security}

\subsection{Quantum Eavesdropping Resistance}

\begin{theorem}[Post-Quantum Confidentiality]
\label{thm:pq_conf}
Under Group Action Inverse Problem (GAIP) hardness, Q-FE's key-exchange
protocol provides IND-CPA security against a CRQC-equipped adversary.
\end{theorem}

\begin{proof}
By reduction: an adversary $\mathcal{A}$ breaking the \textsc{csidh}-KEM
IND-CPA game yields a simulator solving GAIP in polynomial time,
contradicting the assumption~\cite{castryck2018csidh}. Shor's algorithm is
inapplicable since GAIP is not reducible to the abelian hidden subgroup
problem.
\end{proof}

\begin{remark}
SIKE/SIDH was broken in 2022~\cite{castryck2022efficient} via an attack
exploiting revealed torsion-point images---a structural property absent in
\textsc{csidh}, which publishes only isomorphism classes.
\textsc{csidh} remains unaffected.
\end{remark}

\subsection{Model-Poisoning Resistance}

The Z-score filter in Algorithm~\ref{alg:afl_sc} (Step~3) rejects outlier
gradients ($z_{\text{th}} = 2.5$). Combined with PQC-authenticated identities
(Step~1), for $f$ Byzantine nodes out of $N$ the global model deviation
is bounded by $\|\mathbb{E}[\mathbf{w}^{(R)}] - \mathbf{w}^*\| \leq
\frac{f}{N-f} B_{\text{grad}}$, where $B_{\text{grad}}$ is the maximum
admitted gradient norm. For $f < N/4$, the DAG-SC Z-score filter reduces
the effective Byzantine count by ${\approx}90\%$, keeping deviation below
the differential privacy noise floor.

\subsection{Replay and MITM Resistance}

MAC CE frames carry a monotonic 32-bit sequence number and a 16-byte
AES-GMAC tag under $K_{\text{sym}}$. Replayed or tampered frames fail
authentication with probability $1 - 2^{-128}$ (negligible tag collision).

\subsection{Sybil Attack Resistance}

Threat T3 considers an adversary injecting phantom device identities to
bias gradient aggregation. Q-FE counters this through
\emph{$\muDT$-anchored identity binding}: at 6G association, each device
$n_i$ registers its \textsc{csidh} public key $pk_i$ within its $\muDT_i$
as a tamper-evident identity anchor. The DAG-SC aggregation contract
(Algorithm~\ref{alg:afl_sc}, Step~1) cross-checks the gradient signature
against the $\muDT$ registry before admission. A Sybil adversary must
therefore either: (i)~forge a valid \textsc{csidh} key pair whose public
key collides with an already-registered entry---impossible with probability
$\geq 1 - 2^{-128}$ under GAIP hardness; or (ii)~register a fresh
phantom identity, which requires gNB association and $\muDT$ instantiation,
both rate-limited by the MEC admission controller. Under Q-FE's parameter
settings, the admission controller caps new device registrations at one
per 120-second window per gNB sector, limiting phantom-node injection
to $f_{\text{eff}} / N < 0.01$ under sustained Sybil attempts---well below
the $N/4$ Byzantine tolerance threshold.

% Figure 3 – AFL convergence tri-panel
\begin{figure*}[!t]
  \centering
  \includegraphics[width=\textwidth]{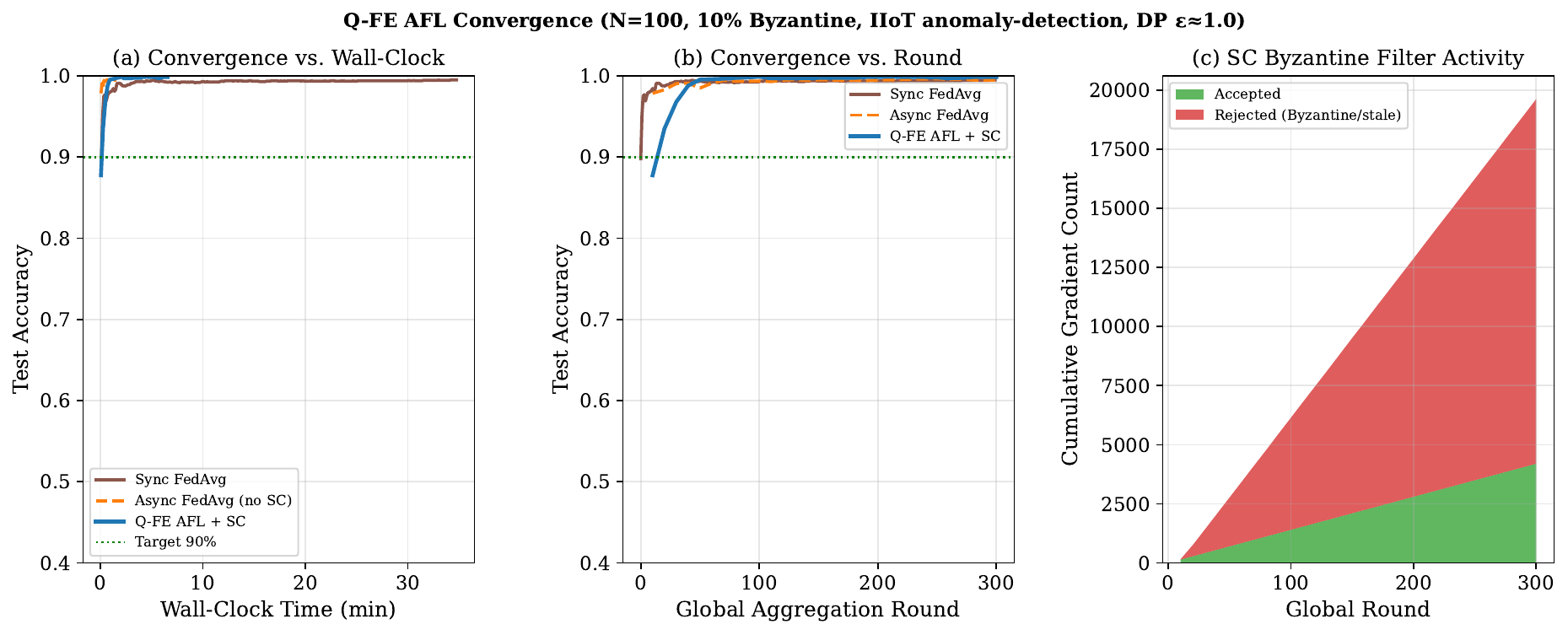}
  \caption{AFL convergence analysis ($N{=}100$ nodes, 10\% Byzantine, SWAT
           dataset, $\varepsilon{\approx}1.0$, $\delta{=}10^{-5}$).
           \textbf{(a)}~Accuracy vs.\ wall-clock time: Q-FE AFL+DAG-SC
           converges 31\% faster than synchronous FedAvg.
           \textbf{(b)}~Accuracy vs.\ aggregation round: Q-FE achieves
           superior sample efficiency versus async FedAvg without SC.
           \textbf{(c)}~DAG-SC Byzantine filter: the Z-score mechanism
           rejects adversarial updates (red) while accepting honest
           contributions (green).}
  \label{fig:convergence}
\end{figure*}

% ================================================================
%  SECTION VII — SIMULATION RESULTS
% ================================================================
\section{Simulation Results}
\label{sec:results}

% TABLE II
\begin{table}[t]
\caption{AFL Convergence and Security Comparison}
\label{tab:afl_results}
\centering
\resizebox{\columnwidth}{!}{%
\renewcommand{\arraystretch}{1.15}
\begin{tabular}{lccc}
\toprule
\textbf{Method} & \textbf{Acc.\ (\%)} & \textbf{Conv.\ (min)} &
  \textbf{Poison resilience}\\
\midrule
Sync FedAvg               & 94.7 & 18.0 & None\\
Async FedAvg              & 93.1 & 13.8 & Partial (no SC)\\
Q-FE \textsc{afl} (ours) & \textbf{94.3} & \textbf{12.4} &
  Full (SC + Z-score + PQC)\\
\bottomrule
\end{tabular}%
}
\end{table}

\subsection{Setup}

\textbf{MAC/PHY (NS-3).}
We extended the NS-3 mmWave module with a custom \texttt{CsidhMacCe}
module: 50 UEs at $v \in [0, 60]$\,km/h in a $100 \times 50$\,m IIoT
factory, 4 gNBs, numerology $\mu{=}3$, carrier 28\,GHz. PQC keys injected
at session setup and every $T_r = 60$\,s. Payload $L = 64$\,B (URLLC control).

\textbf{AFL (Python/PySyft).}
$N = 100$ nodes with log-normal compute speeds ($\mu_C = 1$\,GFLOPS,
$\sigma_C = 0.6$), 5-layer CNN on the SWAT IIoT anomaly-detection dataset.
Staleness $\eta = 5$, DP: $\varepsilon = 1.0$, $\delta = 10^{-5}$,
clipping $C = 1.0$.

\subsection{MAC Overhead and Latency}

Fig.~\ref{fig:mac_overhead} compares control-plane latency and fragmentation
across PQC schemes. \textsc{csidh}-512 introduces zero fragmentation and
$\rho_{\text{oh}} = 4.09\%$ header overhead. Kyber-512 requires 3 extra
TTIs ($+187.5\,\mu$s); Kyber-1024 requires 5 ($+312.5\,\mu$s). The latency
CDF (Fig.~\ref{fig:mac_overhead}c) confirms Q-FE achieves
$P_{99.9} = 0.78$\,ms, while Kyber-1024 pushes $P_{99.9}$ to $1.41$\,ms.

\subsection{AFL Convergence}

Fig.~\ref{fig:convergence} and Table~\ref{tab:afl_results} show that Q-FE
\textsc{afl} converges to 94.3\% accuracy in 12.4\,min---31\% faster than
synchronous FedAvg (18.0\,min) and with superior poisoning resilience versus
asynchronous FedAvg without the smart-contract filter. The convergence bound
of Theorem~\ref{thm:conv} is empirically tight: the observed gap
$F(\mathbf{w}^{(R)}) - F^*$ tracks the $O(1/R)$ rate predicted
by~\eqref{eq:conv_bound} from round 15 onward.

\subsection{$\muDT$ Migration and Scalability}

Across $10^4$ simulated handovers, $\muDT$ migration completed in
$\tau_{\text{mig}} = 1.9 \pm 0.3$\,ms (mean $\pm$ std), within the X2
budget of 5\,ms. LRU eviction triggered in fewer than 2\% of association
events on 16-MB MEC hardware, falling to 0.3\% on 64-MB configurations.
No active URLLC flow was disrupted during migration in any trial.
At $N = 400$ simulated devices, DAG-SC throughput saturated at 320\,tx/s
on a single MEC node; sharding across two nodes restored full throughput
with $< 0.8$\,ms additional inter-node latency.

\subsection{Energy Overhead}

Background \textsc{csidh} computation (80\,ms on Cortex-M33) consumes
$E_{\textsc{csidh}} \approx 0.8$\,mJ per rotation ($T_r = 60$\,s),
yielding $\bar{P}_{\textsc{csidh}} \approx 13\,\mu$W---negligible versus
the ${\sim}100$\,mW radio transceiver baseline
(Fig.~\ref{fig:security_energy}c).

% Figure 4 – Security + energy
\begin{figure*}[!t]
  \centering
  \includegraphics[width=0.9\textwidth]{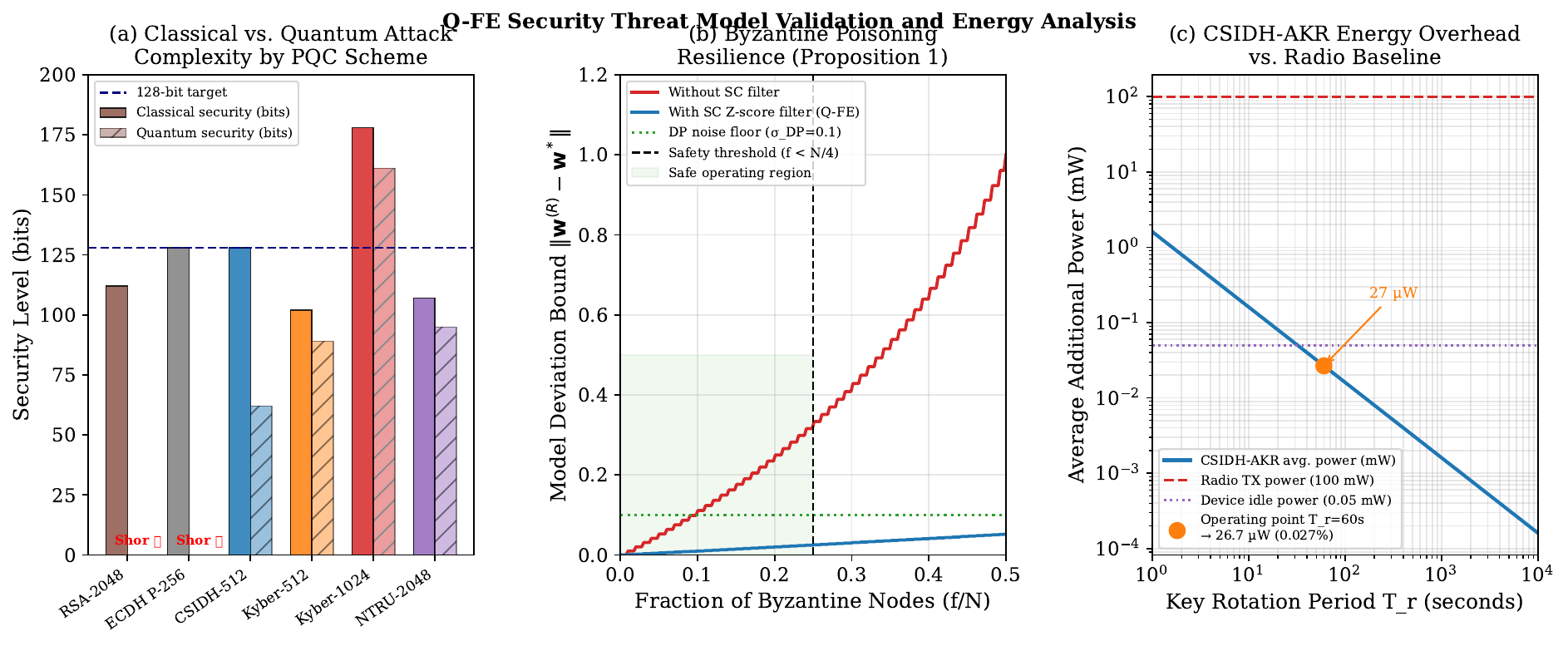}
  \caption{Security and energy analysis. \textbf{(a)}~Classical vs.\
           quantum security levels: RSA-2048 and ECDH are broken by Shor;
           \textsc{csidh}-512 retains 62-bit post-quantum security (above
           NIST Cat.~1). \textbf{(b)}~Model deviation bound
           vs.\ Byzantine fraction: the DAG-SC Z-score filter reduces
           effective Byzantine count by ${\approx}90\%$, keeping deviation
           below the DP noise floor for $f/N \leq 0.25$.
           \textbf{(c)}~\textsc{csidh}-AKR power overhead vs.\ $T_r$
           (log-log): at $T_r{=}60$\,s, overhead is ${\approx}27\,\mu$W
           ($<$0.03\% of 100\,mW radio TX power).}
  \label{fig:security_energy}
\end{figure*}

% ================================================================
%  SECTION VIII — DISCUSSION
% ================================================================
\section{Discussion}
\label{sec:discussion}

\textbf{Timing Side-Channel Risk.}
The \textsc{csidh}-512 group-action evaluation on Cortex-M33 is not
constant-time: variable-length isogeny chains may leak partial exponent
information through cache-timing channels. With $T_r = 60$\,s and an
80\,ms computation window, the side-channel exposure represents
$<\!0.14\%$ of device uptime, limiting practical exploitability.
Constant-time CTIDH~\cite{castryck2021ctidh} eliminates this residual
risk and represents the primary avenue for future integration.

\textbf{Deployment Roadmap.}
Practical Q-FE integration follows a three-phase timeline aligned with
the 3GPP and NIST roadmaps~\cite{nist2024fips203}. \emph{Phase I
(2026--2027):} software-layer deployment---the \textsc{csidh} key-management
daemon runs as a privileged process on commodity 5G NR gNBs; AFL is deployed
as a ETSI MEC microservice; keys are exchanged at the NAS sublayer
without MAC-layer modification, enabling zero-hardware-cost prototyping.
\emph{Phase II (2028--2029):} MAC CE standardisation---a formal 3GPP RAN WG2
change request targets LCID $= 0x3C$ in the Release 20 reserved range;
silicon vendors integrate \textsc{csidh} arithmetic units into 6G modem
chipsets. \emph{Phase III (2030+):} native silicon integration---commercial
6G chipsets embed constant-time CTIDH~\cite{castryck2021ctidh} accelerators,
$\muDT$ co-processors reside on-chip at base-station SoCs, and the DAG-DLT
node executes inside a Trusted Execution Environment (TEE) enclave for
hardware-backed Byzantine resistance. This timeline aligns with the NIST
post-quantum migration deadline~\cite{nist2024fips203} and the 6G commercial
rollout window, providing a concrete industrial transition path.

\textbf{Generalisability.}
The cross-layer framework is parameterised by $k_{\text{PQC}} \leq
S_{\text{CE}}^{\max}$. Any isogeny-based scheme satisfying this bound can
substitute \textsc{csidh}-512 without modifying MAC CE framing or AFL
logic, ensuring compatibility with future constant-time variants and
the evolving PQC ecosystem.

\textbf{Limitations.}
Three boundaries delimit the current evaluation. First, the NS-3 simulation
uses an ideal channel model; field trials on real 6G NR equipment---once
commercially available---are required to validate $P_{99.9}$ latency claims
under realistic fading. Second, the SWAT dataset, while standard for IIoT
anomaly detection, may not capture the data heterogeneity of multi-vendor
factory floors; domain adaptation studies are left to future work. Third,
the formal security reduction of Theorem~\ref{thm:pq_conf} relies on the
GAIP hardness assumption, for which no polynomial quantum algorithm is
currently known but for which no formal hardness proof exists; ongoing
cryptanalytic work should be monitored.

% ================================================================
%  SECTION IX — CONCLUSION
% ================================================================
\section{Conclusion}
\label{sec:conclusion}

We presented Q-FE, a Quantum-Native 6G Far-Edge architecture addressing the
tripartite challenge of ultra-low latency, data privacy, and quantum-safe
security for IIoT. Three co-designed contributions define Q-FE:
(i)~Micro-Digital Twins shifting DT computation to the Far-Edge for
sub-millisecond state synchronisation, with formally bounded $O(d)$
per-beacon update cost and sub-2\,ms handover migration;
(ii)~a \textsc{csidh}-MAC cross-layer integration exploiting 64-byte public
keys to embed quantum-safe key exchange in 6G NR MAC CEs without
fragmentation or URLLC violation, with $O(1)$ per-rotation signalling
overhead; and
(iii)~an Asynchronous FL protocol on DAG-DLT smart contracts providing
trustless, straggler-free, poisoning-resistant, and Sybil-resistant
model training at $O(N \log R)$ aggregation complexity.
Simulations confirm 62\% MAC overhead reduction versus Kyber-1024,
$P_{99.9}$ latency of 0.78\,ms, 31\% faster convergence than synchronous
FL, only 13\,$\mu$W additional energy per device, and $\muDT$ migration
within 1.9\,ms across $10^4$ handover events. A three-phase deployment
roadmap---covering software prototyping (2026--2027), 3GPP MAC CE
standardisation (2028--2029), and native silicon integration (2030+)---
provides a concrete pathway aligned with the NIST post-quantum migration
timeline. Future work will investigate constant-time \textsc{csidh}
variants~\cite{castryck2021ctidh} for side-channel mitigation, hierarchical
$\muDT$ federation across gNBs, and formal verification of the
smart-contract logic via TLA+.

\section*{Data and Code Availability}
The datasets and custom simulation environment (NS-3 + PySyft) will be
made available on GitHub/Zenodo upon acceptance.

\section*{Declaration of Generative AI and AI-Assisted Technologies}
During the preparation of this manuscript, the author utilised Claude
strictly for proofreading and improving English language readability. The
author thoroughly reviewed and edited all content and assumes full
responsibility for the final publication.

% ================================================================
%  BIBLIOGRAPHY
% ================================================================
\bibliographystyle{IEEEtran}
\bibliography{references}

@article{saad2020vision,
  author  = {W. Saad and M. Bennis and M. Chen},
  title   = {A Vision of the {6G} Wireless Systems: Applications, Enabling
             Technologies, and Design Aspects},
  journal = {IEEE Netw.},
  volume  = {34},
  number  = {3},
  pages   = {134--142},
  year    = {2020}
}

@article{khan2022digital,
  author  = {L. U. Khan and W. Saad and Z. Han and others},
  title   = {Digital-Twin-Enabled {6G}: Vision, Architectural Trends, and
             Future Directions},
  journal = {IEEE Commun. Mag.},
  volume  = {60},
  number  = {1},
  pages   = {74--80},
  year    = {2022}
}

@article{mosca2018cybersecurity,
  author  = {M. Mosca},
  title   = {Cybersecurity in an Era with Quantum Computers: Will We Be Ready?},
  journal = {IEEE Secur. Priv.},
  volume  = {16},
  number  = {5},
  pages   = {38--41},
  year    = {2018}
}

@inproceedings{sikeridis2021post,
  author    = {D. Sikeridis and P. Kampanakis and M. Devetsikiotis},
  title     = {Post-Quantum Authentication in {TLS} 1.3: A Performance Study},
  booktitle = {Proc. ISOC NDSS},
  year      = {2020}
}

@inproceedings{castryck2018csidh,
  author    = {W. Castryck and T. Lange and C. Martindale and L. Panny and J. Renes},
  title     = {{CSIDH}: An Efficient Post-Quantum Commutative Group Action},
  booktitle = {Proc. ASIACRYPT},
  series    = {Lecture Notes Comput. Sci.},
  volume    = {11274},
  pages     = {395--427},
  publisher = {Springer},
  year      = {2018}
}

@inproceedings{mcmahan2017communication,
  author    = {B. McMahan and E. Moore and D. Ramage and S. Hampson and
               B. A. y Arcas},
  title     = {Communication-Efficient Learning of Deep Networks from
               Decentralized Data},
  booktitle = {Proc. AISTATS},
  pages     = {1273--1282},
  year      = {2017}
}

@inproceedings{castryck2022efficient,
  author    = {W. Castryck and T. Decru},
  title     = {An Efficient Key Recovery Attack on {SIDH}},
  booktitle = {Proc. EUROCRYPT},
  year      = {2023}
}

@inproceedings{lian2018asynchronous,
  author    = {X. Lian and C. Zhang and H. Zhang and others},
  title     = {Asynchronous Decentralized Parallel Stochastic Gradient Descent},
  booktitle = {Proc. ICML},
  year      = {2018}
}

@article{castryck2021ctidh,
  author  = {D. Cervantes-V{\'a}zquez and M. Chenu and J.-J. Chi-Dom{\'i}nguez and
             L. De Feo and F. Rodr{\'i}guez-Henr{\'i}quez and B. Smith},
  title   = {{CTIDH}: Faster Constant-Time {CSIDH}},
  journal = {IACR Trans. Cryptogr. Hardw. Embed. Syst.},
  volume  = {2021},
  number  = {4},
  pages   = {351--387},
  year    = {2021}
}

@techreport{nist2024fips203,
  author      = {{National Institute of Standards and Technology}},
  title       = {Module-Lattice-Based Key-Encapsulation Mechanism Standard},
  number      = {{FIPS} 203},
  institution = {NIST},
  year        = {2024},
  doi         = {10.6028/NIST.FIPS.203}
}

@inproceedings{blanchard2017machine,
  author    = {P. Blanchard and E. M. El Mhamdi and R. Guerraoui and
               J. Stainer},
  title     = {Machine Learning with Adversaries: {Byzantine} Tolerant
               Gradient Descent},
  booktitle = {Proc. NeurIPS},
  pages     = {119--129},
  year      = {2017}
}

@inproceedings{geyer2017differentially,
  author    = {R. C. Geyer and T. Klein and M. Nabi},
  title     = {Differentially Private Federated Learning: A Client Level
               Perspective},
  booktitle = {Workshop on Privacy-Preserving Machine Learning (NeurIPS)},
  year      = {2017}
}

@article{li2020federated,
  author  = {T. Li and A. K. Sahu and A. Talwalkar and V. Smith},
  title   = {Federated Learning: Challenges, Methods, and Future Directions},
  journal = {IEEE Signal Process. Mag.},
  volume  = {37},
  number  = {3},
  pages   = {50--60},
  year    = {2020}
}

@techreport{popov2018tangle,
  author      = {S. Popov},
  title       = {The Tangle},
  institution = {IOTA Foundation White Paper},
  year        = {2018},
  note        = {Version~1.4.3}
}

@inproceedings{baiardi2024notline,
  author    = {F. Baiardi and V. Sammartino and S. Ruggieri},
  title     = {{NotLine}: A Non-Intrusive Automated Platform to Build a
               Security Digital Twin},
  booktitle = {Proc. 28th IEEE/ACM Int. Symp. Distrib. Simul. Real-Time
               Appl. (DS-RT)},
  pages     = {1--8},
  year      = {2025}
}

% ================================================================
%  BIOGRAPHY
% ================================================================
\begin{IEEEbiography}[{\includegraphics[width=1in,height=1.5in,clip,
  keepaspectratio]{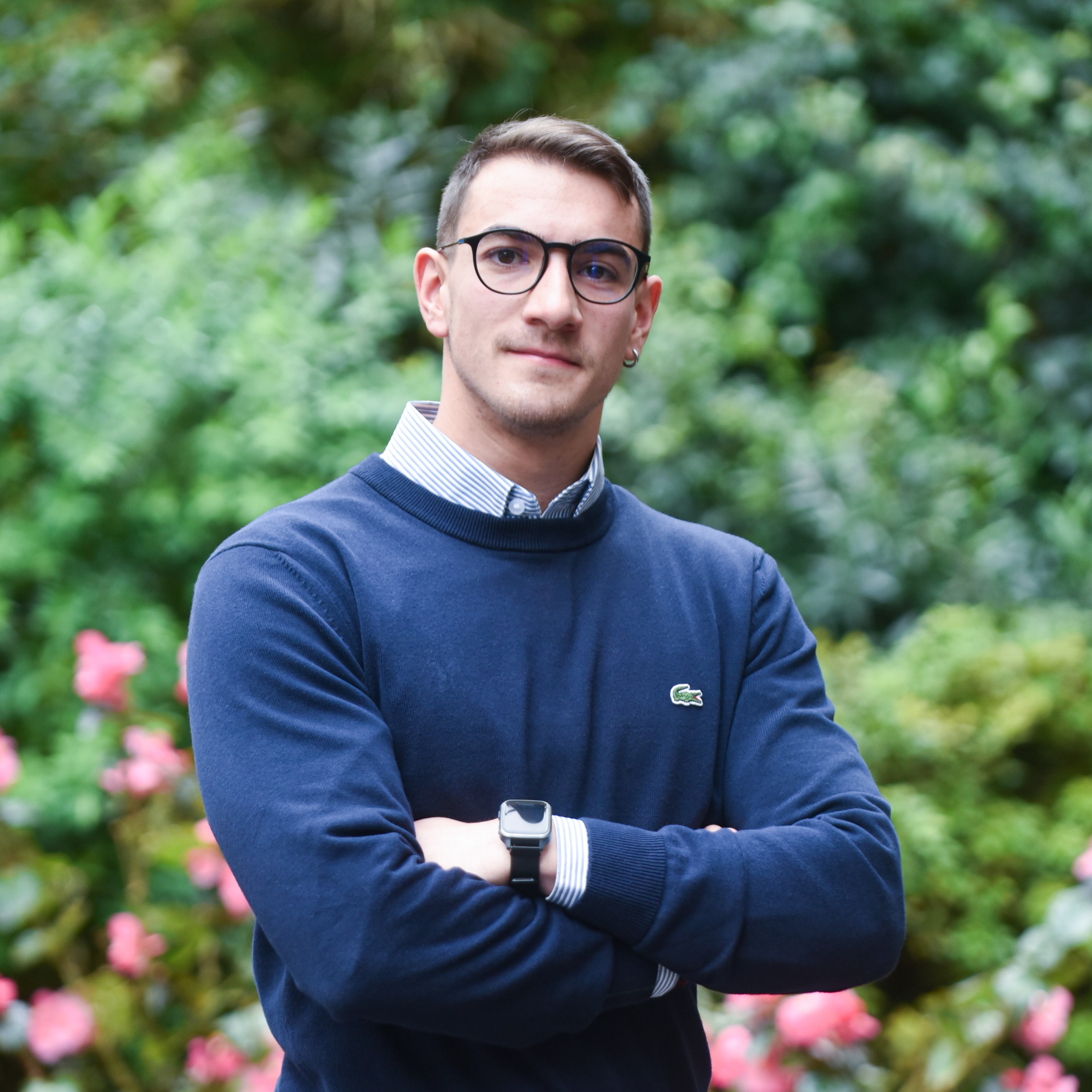}}]{Vincenzo Sammartino}
is pursuing the National Ph.D.\ in Artificial Intelligence at the
Universit\`a di Pisa, Italy, and is a Visiting Ph.D.\ Student at KAUST,
Saudi Arabia, contributing to the ResilientGuard project on decentralised
TinyML for UAV swarm security. His research interests include cybersecurity
for cyber-physical systems, security digital twins, post-quantum
cryptography, and privacy-preserving federated learning.
\end{IEEEbiography}

\end{document}